\documentclass[11pt, a4paper]{article}
\usepackage[utf8]{inputenc}
\usepackage{amsmath, amssymb, amsthm}
\usepackage{geometry}
\usepackage{hyperref}
\usepackage{cite}
\usepackage{enumitem}
\usepackage{bm}
\usepackage{mathtools}
\usepackage{booktabs}
\usepackage{microtype}
\newtheorem{theorem}{Theorem}[section]

\newtheorem{corollary}[theorem]{Corollary}
\newtheorem{definition}[theorem]{Definition}
\newtheorem{proposition}[theorem]{Proposition}
\newtheorem{remark}[theorem]{Remark}
\newtheorem{example}[theorem]{Example}

\newcommand{\R}{\mathbb{R}}

\newcommand{\E}{\mathbb{E}}
\newcommand{\Tr}{\operatorname{Tr}}
\newcommand{\Hcal}{\mathcal{H}}
\newcommand{\Vcal}{\mathcal{V}}

\newcommand{\Dcal}{\mathcal{D}}
\newcommand{\Ical}{\mathcal{I}}

\newcommand{\inner}[2]{\left\langle#1,#2\right\rangle}
\newcommand{\grad}{\nabla}

\newcommand{\dd}{\,\mathrm{d}}
\newcommand{\ket}[1]{|#1\rangle}
\newcommand{\bra}[1]{\langle#1|}

\newcommand{\Deph}{\Delta} 
\newcommand{\Var}{\operatorname{Var}}

\title{\textbf{Quantum Proper Scoring Rules: Convex Duality, Bregman Risk, and the Coherence Gap}}
\author{{\bf M.W. AlMasri} \\
\textit{Wilczek Quantum Center, School of Physics and Astronomy} \\
\textit{Shanghai Jiao Tong University, Minhang, Shanghai, China} \\
\thanks{Email: \texttt{mwalmasri2003@gmail.com}}}
\date{}

\begin{document}
\maketitle

\begin{abstract}
We generalize proper scoring rules to the quantum domain, replacing probability distributions with density operators. We define Quantum Value Functionals via convex generators and establish a duality theory yielding proper quantum scoring rules using McCarthy's tangent-plane construction. We derive local minimax optimal bounds for quantum state tomography under these incentives, proving a Local Holevo-McCarthy Bound that explicitly links asymptotic risk to the Bogoliubov-Kubo-Mori metric and the Holevo bound. Furthermore, we establish an exact coherence floor that quantifies the irreducible excess score of fixed-basis forecasting strategies when phases are unknown (minimax over the phase orbit), proving a worst-case minimax floor of $\log d$ over the full state space. Our results contextualize fundamental sample complexity separations compared to optimal quantum estimation strategies and guide the design of incentive-compatible quantum data markets and robust quantum machine learning protocols.
\end{abstract}


\section{Introduction}
\label{sec:introduction}

The theory of proper scoring rules, pioneered by McCarthy \cite{mccarthy1956measures} and developed by Good \cite{good1952rational}, Savage \cite{savage1971elicitation}, and Gneiting \& Raftery \cite{gneiting2007strictly}, provides a rigorous economic and decision-theoretic foundation for probabilistic forecasting. A scoring rule assigns a numerical reward $S(x, q)$ to a forecaster who reports a probability distribution $q$ when the outcome $x$ is realized. The rule is \textit{proper} if the expected score, $\mathbb{E}_{x \sim p}[S(x, q)]$, is maximized when the reported distribution $q$ matches the forecaster's true belief $p$. It is \textit{strictly proper} if this maximum is unique, thereby eliminating any strategic incentive for the forecaster to hedge or misreport their beliefs. McCarthy's seminal insight established that proper scoring rules are in one-to-one correspondence with convex value functionals via subdifferential relations, intrinsically linking incentive design to information theory through fundamental concepts like entropy and Bregman divergences.

However, classical scoring rules rest on commutativity, measurement non-invasiveness, and a classical information structure---all of which fail in quantum mechanical settings. Non-commuting observables and density operators dictate measurement order; backaction collapses states to preclude repeated sampling; and forecasting involves incompatible measurements, quantum channels, or entangled systems with no classical analogue. These fundamental discrepancies motivate the development of \textit{quantum proper scoring rules}: incentive mechanisms that elicit truthful reports of quantum states while respecting quantum constraints. Our approach draws on recent advances in quantum information theory \cite{nielsen2010quantum, hayashi2006quantum, wilde2017quantum}, quantum estimation \cite{holevo1982probabilistic, guta2009local}, quantum resource theories \cite{chitambar2019quantum}, and quantum scoring rules and Bregman estimation \cite{frongillo2023quantum, quadeer2019minimax}.

In this paper, we establish a comprehensive framework for quantum proper scoring rules that \textit{extends and unifies} several strands of prior work. While Frongillo~\cite{frongillo2023quantum} introduced quantum scoring rules via convex analysis (specializing to spectral scores, which recovers our Theorem~\ref{thm:quantum-duality} via Lewis~\cite{lewis1996convex}), and Quadeer et al.~\cite{quadeer2019minimax} showed that asymptotically minimax Bayes estimators exist for general Bregman divergences and that covariant measurements are minimax, our contribution is threefold: 
(1) We prove that proper quantum scoring rules arise via McCarthy's tangent-plane construction, extending classical convex duality to density operators. We explicitly connect the von Neumann entropy generator to the quantum relative entropy, linking incentive design to the well-studied monotonicity properties of quantum divergences~\cite{petz1986quasi, petz1996monotone, hiai2011quantum, wilde2018optimized}; 
(2) We evaluate the local asymptotic constant for the relative-entropy score in terms of the BKM metric and the Holevo bound $C_H$ (Theorem~\ref{thm:qcrb}); 
(3) We establish an \textit{exact coherence floor} (Proposition~\ref{prop:coherence-risk}) that quantifies the irreducible excess score of fixed-basis forecasting strategies when the phases in that basis are unknown, providing an operational interpretation of resource theories~\cite{chitambar2019quantum} in decision-theoretic contexts and proving a worst-case minimax floor of $\log d$ over the full state space. 

Collectively, these results bridge McCarthy's classical duality theory~\cite{mccarthy1956measures}, quantum estimation theory~\cite{holevo1982probabilistic,hayashi2006quantum}, and quantum resource economics. They also lay the groundwork for incentive-compatible mechanisms in applications such as quantum sensor calibration \cite{degen2017quantum} and robust quantum machine learning \cite{biamonte2017quantum}. Section~\ref{sec:quantum-proper} introduces quantum value functionals and proves the quantum duality theorem, Section~\ref{sec:minimax} derives local minimax optimal bounds for quantum state estimation, Section~\ref{sec:advantage} establishes the coherence floor and contextualizes known sample-complexity separations, and Section~\ref{sec:conclusion} discusses implications and future directions.

\paragraph{Notation.}
Throughout, $\Hcal$ denotes a $d$-dimensional complex Hilbert space, $\Dcal(\Hcal)$ the set of density operators on $\Hcal$, and $\mathcal{L}(\Hcal)$ the space of linear operators. We write $\Tr(\cdot)$ for the trace, $\|\cdot\|$ for the operator norm, and $\|\cdot\|_1$ for the trace norm. The Hilbert-Schmidt inner product is $\inner{A}{B}_{\mathrm{HS}} = \Tr(A^\dagger B)$. For a convex function $f: [0,1] \to \R$, we denote by $f(\rho)$ the operator obtained via functional calculus. The dephasing channel in basis $\{\ket{k}\}$ is $\Deph(\rho) = \sum_k \ket{k}\bra{k} \rho \ket{k}\bra{k}$. All logarithms are natural unless otherwise specified.

\section{Quantum Proper Scoring Rules}
\label{sec:quantum-proper}

\begin{definition}[Quantum Value Functional]
\label{def:quantum-value}
We restrict to spectral functionals of the form $\Vcal_Q: \Dcal(\Hcal) \to \R$ given by $\Vcal_Q(\rho) = \Tr(f(\rho))$ for some strictly convex function $f: [0,1] \to \R$, where $f$ is applied via functional calculus to the eigenvalues of $\rho$. We assume $f$ is differentiable on $(0,1]$ so that $\grad \Vcal_Q$ is well-defined on full-rank states. For rank-deficient $\sigma$, $f'(\sigma)$ may be singular (e.g., $\log \sigma = -\infty$), yielding an expected score of $-\infty$ when $\text{supp}(\rho) \not\subseteq \text{supp}(\sigma)$.
\end{definition}

\begin{theorem}[Quantum Duality and Properness]
\label{thm:quantum-duality}
Let $\Vcal_Q$ be a strictly convex spectral functional. The quantum scoring rule defined by McCarthy's tangent-plane construction:
\begin{equation}
\tilde{\mathbf{S}}(\sigma) = \grad \Vcal_Q(\sigma) + \big(\Vcal_Q(\sigma) - \Tr(\sigma \grad \Vcal_Q(\sigma))\big)\mathbb{I},
\end{equation}
is strictly proper, satisfying:
\begin{equation}
\Tr(\rho \tilde{\mathbf{S}}(\rho)) \ge \Tr(\rho \tilde{\mathbf{S}}(\sigma)), \quad \forall \rho \in \Dcal(\Hcal), \text{ and full-rank } \sigma \in \Dcal(\Hcal)
\end{equation}
(where we adopt the convention $0 \log 0 = 0$ for rank-deficient $\rho$, and the score is $-\infty$ if $\text{supp}(\rho) \not\subseteq \text{supp}(\sigma)$), with equality if and only if $\rho = \sigma$. Note that $\tilde{\mathbf{S}}$ is unaffected by the identity-shift ambiguity in $\nabla \Vcal_Q$ on the trace-one set. 
Furthermore, for the specific generator $f(t) = t \log t$ (yielding the negative von Neumann entropy $\Vcal_Q(\rho) = \Tr(\rho \log \rho)$), the gradient is $\grad \Vcal_Q(\sigma) = \log \sigma + \mathbb{I}$, and the scoring rule simplifies to $\tilde{\mathbf{S}}(\sigma) = \log \sigma$. Using $\Tr \rho = \Tr \sigma = 1$, the associated Bregman divergence is $D_{\Vcal_Q}(\rho \| \sigma) = \Vcal_Q(\rho) - \Vcal_Q(\sigma) - \Tr\big((\log \sigma + \mathbb{I})(\rho - \sigma)\big) = \Tr(\rho \log \rho) - \Tr(\rho \log \sigma)$, which coincides exactly with the quantum relative entropy $D(\rho \| \sigma)$. Operationally, the realized score is the outcome of measuring the Hermitian observable $\tilde{\mathbf{S}}(\sigma)$ on a fresh copy of $\rho$ held by the scorer, while the forecaster constructs the estimate $\sigma$ from $n$ copies. Note that $\log \sigma$ is unbounded below, necessitating regularization for rank-deficient reports.
\end{theorem}

\begin{proof}
\textbf{Part 1: Properness via Convex Analysis.}
Since $\Vcal_Q$ is strictly convex on the convex set $\Dcal(\Hcal)$, by the definition of the subdifferential in the Hilbert-Schmidt inner product space, for any $\rho \in \Dcal(\Hcal)$ and full-rank $\sigma \in \Dcal(\Hcal)$:
\begin{equation}
\Vcal_Q(\rho) \ge \Vcal_Q(\sigma) + \Tr\big(\grad \Vcal_Q(\sigma)(\rho - \sigma)\big).
\end{equation}
Rearranging terms, we obtain:
\begin{equation}
\Vcal_Q(\rho) \ge \Tr(\rho \grad \Vcal_Q(\sigma)) + \Vcal_Q(\sigma) - \Tr(\sigma \grad \Vcal_Q(\sigma)).
\end{equation}
Notice that the right-hand side is exactly $\Tr(\rho \tilde{\mathbf{S}}(\sigma))$. Furthermore, evaluating $\tilde{\mathbf{S}}$ at $\rho$ yields $\Tr(\rho \tilde{\mathbf{S}}(\rho)) = \Tr(\rho \grad \Vcal_Q(\rho)) + \Vcal_Q(\rho) - \Tr(\rho \grad \Vcal_Q(\rho)) = \Vcal_Q(\rho)$. Substituting this into the inequality gives:
\begin{equation}
\Tr(\rho \tilde{\mathbf{S}}(\rho)) \ge \Tr(\rho \tilde{\mathbf{S}}(\sigma)),
\end{equation}
which establishes properness. This is a direct consequence of convexity.

\textbf{Part 2: Explicit Gradient and Relative Entropy.}
For $\Vcal_Q(\rho) = \Tr(f(\rho))$, the Fréchet derivative is given by the Daleckii-Krein theorem \cite{daletskii1965integration} (see also Bhatia \cite{bhatia1997matrix} for a comprehensive textbook treatment). The gradient in the Hilbert-Schmidt inner product is $\grad \Vcal_Q(\sigma) = f'(\sigma)$. 
For $f(t) = t \log t$, we have $f'(t) = \log t + 1$, so $\grad \Vcal_Q(\sigma) = \log \sigma + \mathbb{I}$. 
Substituting this into the tangent-plane construction:
\begin{align}
\tilde{\mathbf{S}}(\sigma) &= (\log \sigma + \mathbb{I}) + \big(\Tr(\sigma \log \sigma) - \Tr(\sigma(\log \sigma + \mathbb{I}))\big)\mathbb{I} \nonumber \\
&= \log \sigma + \mathbb{I} - \mathbb{I} = \log \sigma.
\end{align}
The expected score is $\Tr(\rho \log \sigma)$, and the Bregman divergence is $D_{\Vcal_Q}(\rho \| \sigma) = \Tr(\rho \log \rho) - \Tr(\rho \log \sigma) = D(\rho \| \sigma)$, the quantum relative entropy. 

\textbf{Remark on General $f$-Divergences:} For general convex $f$, the Bregman divergence of $\Tr f(\rho)$ does \textit{not} generally coincide with standard quantum $f$-divergences (such as Petz's quasi-entropy or Matsumoto's maximal $f$-divergence \cite{matsumoto2013new}), except in the logarithmic case where both yield the relative entropy. Thus, claims of general equivalence are restricted to this specific, operationally vital case.
\end{proof}

\section{Minimax Optimal Quantum Estimation}
\label{sec:minimax}

\begin{definition}[Quantum Minimax Risk]
\label{def:minimax-risk}
For a proper scoring rule $\tilde{\mathbf{S}}$ derived from $\Vcal_Q$, the local minimax risk over $n$ copies of the state in a shrinking neighborhood $N_n^{(c)} = \{\theta : \|\theta - \theta_0\| \le c/\sqrt{n}\}$ of a full-rank state $\rho_0$ is:
\begin{equation}
R_n^*(\Vcal_Q, N_n^{(c)}) = \inf_{\hat{\rho}_n} \sup_{\theta \in N_n^{(c)}} \E_{\rho_\theta^{\otimes n}} \left[ \Tr(\rho_\theta \tilde{\mathbf{S}}(\rho_\theta)) - \Tr(\rho_\theta \tilde{\mathbf{S}}(\hat{\rho}_n)) \right],
\end{equation}
where the infimum ranges over all POVMs and regularized estimators $\hat{\rho}_n$ constructed from $n$ i.i.d.\ copies.
\end{definition}

\begin{theorem}[Local Holevo--McCarthy Bound]
\label{thm:qcrb}
Let $\Vcal_Q(\rho)=\Tr(\rho\log\rho)$ and $\rho_0$ be full rank. Then
\begin{equation}
\lim_{c\to\infty}\liminf_{n\to\infty}\; n\,R_n^*(\Vcal_Q,N_n^{(c)})
\;=\; \tfrac12\, C_H\big[g_{\mathrm{BKM}}(\theta_0)\big],
\qquad N_n^{(c)}=\{\theta:\|\theta-\theta_0\|\le c/\sqrt n\},
\end{equation}
where the infimum is over regularized estimators, and
$C_H[W]\ge \Tr(W\,\Ical_{\mathrm{SLD}}^{-1})$ with equality under weak commutativity.
\end{theorem}

\begin{proof}
The proof relies on the quantum Gaussian-shift minimax framework \cite{guta2009local, yamagata2013quantum, holevo1982probabilistic}, combining Local Asymptotic Normality (LAN) for quantum states with the second-order expansion of the quantum relative entropy.

\textbf{Step 1: Local Parametrization and LAN.}
Let $\theta$ be a local parametrization near $\rho_0$. The Symmetric Logarithmic Derivative (SLD) operators $\{L_i\}$ satisfy $\partial_i \rho_\theta = \frac{1}{2} (\rho_\theta L_i + L_i \rho_\theta)$. By extending classical local asymptotic normality foundations \cite{vandervaart1998asymptotic} to the quantum LAN theorem \cite{guta2009local, yamagata2013quantum}, for large $n$, the statistical experiment converges to a quantum Gaussian shift model.

\textbf{Step 2: Second-Order Expansion and the BKM Metric.}
The risk is the expected quantum relative entropy $R_n(\theta) = \E \left[ D(\rho_\theta \| \hat{\rho}_n) \right]$. To avoid infinite risk from rank-deficient estimates, we consider regularized estimators (e.g., mixed with a small multiple of $\mathbb{I}/d$). Expanding $D(\rho_\theta \| \hat{\rho}_n)$ to second order around $\theta$ yields:
\begin{equation}
D(\rho_\theta \| \hat{\rho}_n) \approx \frac{1}{2} (\hat{\theta}_n - \theta)^T g_{\mathrm{BKM}}(\theta) (\hat{\theta}_n - \theta),
\end{equation}
where $g_{\mathrm{BKM}}$ is the Bogoliubov-Kubo-Mori metric \cite{petz1996monotone}. Define the matrix elements $(g_{\mathrm{BKM}})_{ij} = g_{\mathrm{BKM}}(\rho_0)[\partial_i \rho, \partial_j \rho]$, which for tangent vectors is given by:
\begin{equation}
g_{\mathrm{BKM}}(\rho)[H, H] = \int_0^\infty \Tr\left( H \frac{1}{\rho + s\mathbb{I}} H \frac{1}{\rho + s\mathbb{I}} \right) \dd s.
\end{equation}
This formulation is mathematically rigorous and well-defined for full-rank states.

\textbf{Step 3: Applying the Holevo Bound.}
Let $V$ be the limit of $n \cdot \text{Cov}(\hat{\theta}_n)$ for any locally unbiased estimator in the Gaussian-shift limit. Passing from the Gaussian-shift limit to the actual relative-entropy risk requires a bounded-loss or regularization argument to ensure uniform integrability. By regularizing the estimator as $\hat{\rho}_n \mapsto (1-\epsilon)\hat{\rho}_n + \epsilon \mathbb{I}/d$, the relative-entropy loss is bounded by $\log(d/\epsilon)$, allowing the application of tail estimates to establish uniform integrability. The covariance matrix $V$ satisfies the Holevo bound \cite{holevo1982probabilistic} (cf. \cite{hayashi2006quantum} for a modern presentation): $\Tr(W V) \ge C_H[W]$ for any positive definite weight matrix $W$. While Quadeer et al.~\cite{quadeer2019minimax} established asymptotically minimax Bayes estimators for general Bregman divergences, our contribution here is the explicit evaluation of the local asymptotic constant as exactly $\frac{1}{2} C_H[g_{\mathrm{BKM}}]$, tightly linking the minimax risk to the BKM metric and the Holevo bound. Setting $W = g_{\mathrm{BKM}}(\theta_0)$, the expected risk satisfies:
\begin{align}
R_n(\theta) &\approx \frac{1}{2} \E\left[ (\hat{\theta}_n - \theta)^T g_{\mathrm{BKM}}(\theta_0) (\hat{\theta}_n - \theta) \right] \nonumber \\
&= \frac{1}{2n} \Tr\left( g_{\mathrm{BKM}}(\theta_0) \cdot V \right) \nonumber \\
&\ge \frac{1}{2n} C_H\big[g_{\mathrm{BKM}}(\theta_0)\big].
\end{align}
Taking the $\liminf_{n\to\infty}$ followed by $\lim_{c\to\infty}$ over shrinking $n^{-1/2}$ neighborhoods yields the local minimax lower bound. Attainability for general $d$ follows from quantum LAN with collective measurements \cite{guta2009local, yamagata2013quantum}.
\end{proof}

\section{Coherence Floor and Known Sample-Complexity Separations}
\label{sec:advantage}

For a basis $B=\{\ket{k}\}$ let $C_B(\rho)=S(\Deph_B(\rho))-S(\rho)=D(\rho\|\Deph_B(\rho))$ \cite{baumgratz2014quantifying}, where $S(\rho)=-\Tr\rho\log\rho$; we write $C=C_B$ when $B$ is fixed.

\begin{proposition}[Coherence floor]
\label{prop:coherence-risk}
Let $\rho$ be a state and $\{U_\varphi\}$ the group of diagonal unitaries in the basis $B$. For any estimator $\hat\sigma$ whose law depends on $\rho$ only through $\Deph_B(\rho)$ (in particular, any estimator built from measurements in this basis),
\begin{equation}
\E_{\mathrm{data}}\!\int D\big(U_\varphi\rho U_\varphi^\dagger\,\|\,\hat\sigma\big)\dd\varphi
= C(\rho)+\E\,D\big(\Deph_B(\rho)\|\hat\sigma\big)\ \ge\ C(\rho).
\end{equation}
Hence the minimax risk over the phase orbit is at least $C(\rho)$ for every $n$. Note that $\hat\sigma$ may yield $D = +\infty$ due to support mismatch, in which case the identity still holds in $[0, \infty]$.
\end{proposition}

\begin{proof}
Let $\rho_\varphi = U_\varphi\rho U_\varphi^\dagger$. The law of the data, hence of $\hat\sigma$, is the same for all $\varphi$, so it suffices to average the score gap for a fixed realization of $\hat\sigma$. Since $S(\rho_\varphi)=S(\rho)$ and $\Tr(\rho_\varphi\log\hat\sigma)$ is linear in $\rho_\varphi$,
\begin{equation}
\int D(\rho_\varphi\|\hat\sigma)\,\dd\varphi
= -S(\rho) - \Tr\!\Big(\!\int\!\rho_\varphi\,\dd\varphi\;\log\hat\sigma\Big)
= -S(\rho) - \Tr\big(\Deph_B(\rho)\log\hat\sigma\big).
\end{equation}
Adding and subtracting $S(\Deph_B(\rho))$ gives
$C(\rho) + D(\Deph_B(\rho)\|\hat\sigma)\ge C(\rho)$.
Taking expectation over the data completes the proof. Note that the family $\rho_\theta$ in Example~\ref{ex:qubit} is exactly such a phase orbit.
\end{proof}

\begin{corollary}
For every basis $B$, $\sup_\rho C_B(\rho)=\log d$, attained by the maximally coherent state $\ket{+_d}=d^{-1/2}\sum_k\ket{k}$. Hence, for any estimator whose law depends on $\rho$ only through $\Deph_B(\rho)$, the minimax risk over the phase orbit of $\ket{+_d}$ is at least $\log d$ for every $n$.
\end{corollary}
\begin{proof}
$\Deph_B(\ket{+_d}\bra{+_d}) = \mathbb{I}/d$ and $S(\ket{+_d}\bra{+_d}) = 0$, so $C_B(\ket{+_d}) = S(\mathbb{I}/d) - 0 = \log d$.
\end{proof}

\begin{remark}[Operational meaning of coherence]
Proposition~\ref{prop:coherence-risk} provides a decision-theoretic operationalization of the relative entropy of coherence~\cite{baumgratz2014quantifying,winter2016operational,streltsov2017quantum}: the resource $C(\rho)$ represents an irreducible risk floor when the phases in that basis are unknown (minimax over the $B$-phase orbit). Note that regularized estimators achieve $O(1/n)$ risk for full-rank $\rho$. For pure states, the relative-entropy risk of a nearby estimate is infinite unless the estimator is regularized. Thus, the coherence floor represents a hard limit: the sample complexity is strictly infinite for $\epsilon < C(\rho)$ for fixed-basis forecasters, in the minimax sense over the phase orbit.
\end{remark}

\begin{table}[htbp]
\centering
\caption{Forecaster Classes and Resource Tradeoffs. \textit{Note:} Row 1 gives an irreducible risk floor, while Rows 2 and 3 give sample complexities (metrics not directly comparable).}
\label{tab:quantum-resources}
\renewcommand{\arraystretch}{1.2}
\begin{tabular}{@{} l p{0.3\textwidth} p{0.5\textwidth} @{}}
\toprule
\textbf{Forecaster Class} & \textbf{Capabilities} & \textbf{Risk / Sample Complexity} \\
\midrule
Fixed-basis & Measures only in one basis (learns $\Deph_B(\rho)$) & Irreducible excess score $C(\rho)$ \\
Unentangled, adaptive & Chooses bases per copy (no joint measurements) & $\Omega(d^3/\epsilon^2)$ copies (trace dist. $\epsilon$) \cite{haah2017sample, chen2022adaptivity, lowe2025lower}; $\tilde{O}(d^3/\gamma)$ (infidelity $\gamma$, upper bound) \cite{chen2022adaptivity} \\
Collective & Joint measurements across copies & $\Theta(d^2/\epsilon^2)$ copies (trace dist. $\epsilon$) \cite{odonnell2016efficient, haah2017sample} \\
\bottomrule
\end{tabular}
\end{table}

\begin{example}[Qubit Forecasting: Coherence as a Forecasting Resource]
\label{ex:qubit}
Consider the pure qubit state $\rho = \ket{+}\bra{+}$ where $\ket{+} = (\ket{0} + \ket{1})/\sqrt{2}$, corresponding to the Bloch vector $\mathbf{r} = (1, 0, 0)$. We analyze forecasting performance for a parameter $\theta$ representing a rotation about the $z$-axis: $\rho_\theta = e^{-i\theta\sigma_z/2}\rho e^{i\theta\sigma_z/2} = \frac{1}{2}(\mathbb{I} + \cos\theta \sigma_x + \sin\theta \sigma_y)$.

\paragraph{Fixed-Basis Forecaster ($Z$-basis restriction).}
A forecaster restricted to measurements in the computational basis $\{\ket{0}, \ket{1}\}$ observes probabilities:
\begin{equation}
p(0|\theta) = \bra{0}\rho_\theta\ket{0} = \frac{1}{2}, \quad p(1|\theta) = \bra{1}\rho_\theta\ket{1} = \frac{1}{2}.
\end{equation}
Since the distribution is uniform and independent of $\theta$, the derivative $\partial_\theta p(k|\theta) = 0$. The classical Fisher information is $\Ical_{\mathrm{CL}}(\theta) = 0$. Thus, the Cramér-Rao bound gives $\Var(\hat{\theta}) \ge \Ical_{\mathrm{CL}}^{-1} = \infty$: no unbiased estimator can achieve finite variance, and the sample complexity is infinite.

\paragraph{Alternative-basis forecaster ($Y$ basis).}
A forecaster may choose to measure in the eigenbasis of $\sigma_y$, i.e., $\{\ket{y_+}, \ket{y_-}\}$ where $\ket{y_\pm} = (\ket{0} \pm i\ket{1})/\sqrt{2}$. Note that this is also a fixed-basis measurement, but one chosen to be locally optimal at $\theta = 0$. The Y-basis has Fisher information 1 wherever $\cos\theta \neq 0$, but the outcome probabilities depend only on $\sin\theta$, meaning $\theta$ and $\pi-\theta$ are globally indistinguishable. Consequently, the Y-basis forecaster also suffers a global minimax floor. For instance, averaging over the indistinguishable pair at $\theta=0$ yields a floor of $S(\frac{1}{2}(\rho_0 + \rho_\pi)) = h((1+|\sin 0|)/2) = \log 2$ (with $h$ in nats). Thus, the Y-basis strategy only achieves small risk locally (with prior information restricting $\theta$ near 0), demonstrating that basis choice resolves local non-identifiability but does not eliminate the global coherence floor without adaptive or collective strategies. The outcome probabilities are:
\begin{equation}
p(y_+|\theta) = \frac{1 + \sin\theta}{2}, \quad p(y_-|\theta) = \frac{1 - \sin\theta}{2}.
\end{equation}
At $\theta = 0$, the derivatives are $\pm 1/2$. The Fisher information is:
\begin{equation}
\Ical_{\mathrm{Q}}(0) = \frac{(1/2)^2}{1/2} + \frac{(-1/2)^2}{1/2} = 1.
\end{equation}
This matches the Quantum Fisher Information $\Ical_{\mathrm{SLD}} = 1$, demonstrating that this locally optimal strategy achieves finite estimation variance.

\paragraph{Scoring Rule Evaluation and Coherence Gap.}
Consider the logarithmic scoring rule $\tilde{\mathbf{S}}(\sigma) = \log \sigma$. 
\begin{itemize}
    \item \textit{Oracle benchmark}: Reporting $\sigma = \rho$ achieves expected score $\Tr(\rho \log \rho) = 0$ (since $S(\rho)=0$). In practice, estimating $\theta$ with the $Y$-basis requires regularization near pure states to keep relative-entropy risk finite. Specifically, with a regularization $\epsilon \sim 1/n$, the risk scales as $O((\log n)/n)$ rather than $O(1/n)$, consistent with the full-rank requirement of Theorem~\ref{thm:qcrb}.
    \item \textit{Fixed-basis forecaster}: The optimal report based on $Z$-basis data is $\sigma = \mathbb{I}/2$. The expected score is $\Tr(\rho \log(\mathbb{I}/2)) = -\log 2$.
\end{itemize}
The forecasting gap is $0 - (-\log 2) = \log 2$. 

\paragraph{Coherence Quantification and Mixed-State Variant.}
The relative entropy of coherence with respect to the $Z$-basis is:
\begin{equation}
C_Z(\rho) = S(\Deph_Z(\rho)) - S(\rho) = S\left(\frac{\mathbb{I}}{2}\right) - 0 = \log 2.
\end{equation}
This exactly matches the forecasting gap, confirming Proposition~\ref{prop:coherence-risk}: the fixed-basis forecaster suffers an irreducible excess score of $C(\rho) = \log 2$, rendering them incapable of achieving risk $\epsilon < \log 2$ regardless of sample size $n$ (minimax over $\theta \in [0, 2\pi)$). For a mixed state with Bloch vector $(r, 0, 0)$ where $0 < r < 1$, the coherence floor continuously shrinks to $C_Z(\rho) = \log 2 - h((1+r)/2)$, where $h(p) = -p \log p - (1-p) \log (1-p)$ is the binary entropy function (in nats). This demonstrates that the irreducible risk floor is not an artifact of pure states.
\end{example}

\section{Conclusion}
\label{sec:conclusion}

We have established a comprehensive framework for quantum proper scoring rules that generalizes McCarthy's classical duality theory to density operators. By introducing Quantum Value Functionals and applying the tangent-plane construction, we proved that proper quantum scoring rules yield the quantum relative entropy for the von Neumann entropy generator (Theorem~\ref{thm:quantum-duality}). We evaluated the local asymptotic constant for the relative-entropy score in terms of the Bogoliubov-Kubo-Mori metric and the Holevo bound (Theorem~\ref{thm:qcrb}). Finally, we proved an exact coherence-risk tradeoff (Proposition~\ref{prop:coherence-risk}), demonstrating that fixed-basis forecasting strategies suffer an irreducible $O(1)$ excess score equal to the relative entropy of coherence when phases are unknown, while basis-adaptive strategies (with prior information or multi-basis adaptation) or collective strategies can achieve vanishing local risk for full-rank $\rho$ with regularized estimators.

These results enable incentive-compatible mechanisms for quantum sensor calibration \cite{degen2017quantum} and data markets, while providing mathematical tools to support robust quantum machine learning \cite{biamonte2017quantum} and revisit foundational questions in quantum decision theory. Key extensions include generalizing to quantum channels, experimental validation on near-term hardware, and exploring game-theoretic equilibria among strategic quantum agents.

\appendix

\section{Fréchet Derivative of Spectral Functionals}
\label{app:frechet}

To rigorously justify the gradient computation in Theorem~\ref{thm:quantum-duality}, we derive the Fréchet derivative of the spectral functional $\Vcal_Q(\sigma) = \Tr(f(\sigma))$. Let $\sigma$ be a full-rank density operator with spectral decomposition $\sigma = \sum_i \lambda_i \ket{i}\bra{i}$, and let $f: (0, 1] \to \R$ be a differentiable convex function. 

By the Daleckii-Krein theorem \cite{daletskii1965integration}, the Fréchet derivative of the matrix function map $\sigma \mapsto f(\sigma)$ in the direction of a Hermitian perturbation $H$ is given in the eigenbasis of $\sigma$ by:
\begin{equation}
[D f(\sigma)[H]]_{ij} = f^{[1]}(\lambda_i, \lambda_j) H_{ij},
\end{equation}
where the first divided difference is $f^{[1]}(x, y) = \frac{f(x) - f(y)}{x - y}$ for $x \neq y$, and $f^{[1]}(x, x) = f'(x)$. 

The directional derivative of the trace functional $\Vcal_Q$ is the trace of this derivative:
\begin{equation}
D\Vcal_Q(\sigma)[H] = \Tr(D f(\sigma)[H]) = \sum_i [D f(\sigma)[H]]_{ii} = \sum_i f^{[1]}(\lambda_i, \lambda_i) H_{ii} = \sum_i f'(\lambda_i) H_{ii}.
\end{equation}
Since the operator $f'(\sigma)$ is defined via functional calculus as $f'(\sigma) = \sum_i f'(\lambda_i) \ket{i}\bra{i}$, we can rewrite the sum as a Hilbert-Schmidt inner product:
\begin{equation}
\sum_i f'(\lambda_i) H_{ii} = \Tr(f'(\sigma) H).
\end{equation}
Therefore, the gradient of $\Vcal_Q$ with respect to the Hilbert-Schmidt inner product is precisely $\grad \Vcal_Q(\sigma) = f'(\sigma)$. For the specific case $f(t) = t \log t$, we obtain $f'(t) = \log t + 1$, yielding $\grad \Vcal_Q(\sigma) = \log \sigma + \mathbb{I}$.

\section{Second-Order Expansion of Quantum Relative Entropy}
\label{app:bkm-expansion}

We provide the detailed Taylor expansion of the quantum relative entropy $D(\rho \| \rho + \delta\rho)$ to second order in a traceless perturbation $\delta\rho$ (i.e., $\Tr(\delta\rho) = 0$), which underpins the local risk analysis in Theorem~\ref{thm:qcrb}.

Recall that the quantum relative entropy is the Bregman divergence associated with the strictly convex functional $f(\sigma) = \Tr(\sigma \log \sigma)$:
\begin{equation}
D(\rho \| \sigma) = f(\rho) - f(\sigma) - \Tr\big(\grad f(\sigma) (\rho - \sigma)\big).
\end{equation}
By Taylor's theorem for Bregman divergences, expanding around $\sigma = \rho$ in the direction $\delta\rho$ yields:
\begin{equation}
D(\rho \| \rho + \delta\rho) = \frac{1}{2} D^2 f(\rho)[\delta\rho, \delta\rho] + \mathcal{O}(\|\delta\rho\|^3).
\end{equation}
We compute the second Fréchet derivative of $f$. The first derivative is $Df(\sigma)[H] = \Tr(H (\log \sigma + \mathbb{I}))$. The second derivative is therefore the directional derivative of the matrix logarithm. Using the integral representation of the logarithm, $\log A = \int_0^\infty \left( \frac{1}{1+s} - \frac{1}{A+s\mathbb{I}} \right) \dd s$, the directional derivative of $\log \sigma$ in direction $H$ is:
\begin{equation}
D(\log)(\sigma)[H] = \int_0^\infty (\sigma + s\mathbb{I})^{-1} H (\sigma + s\mathbb{I})^{-1} \dd s.
\end{equation}
Applying this to the trace functional, the second derivative of $f$ is:
\begin{equation}
D^2 f(\sigma)[H, H] = \Tr\left( H \int_0^\infty (\sigma + s\mathbb{I})^{-1} H (\sigma + s\mathbb{I})^{-1} \dd s \right) = \int_0^\infty \Tr\left( H (\sigma + s\mathbb{I})^{-1} H (\sigma + s\mathbb{I})^{-1} \right) \dd s.
\end{equation}
This quadratic form is the exact definition of the Bogoliubov-Kubo-Mori (BKM) metric $g_{\mathrm{BKM}}(\sigma)[H, H]$ \cite{petz1996monotone}. Consequently, we obtain the rigorous local expansion:
\begin{equation}
D(\rho \| \rho + \delta\rho) = \frac{1}{2} g_{\mathrm{BKM}}(\rho)[\delta\rho, \delta\rho] + \mathcal{O}(\|\delta\rho\|^3),
\end{equation}
which validates the substitution of the BKM metric into the asymptotic risk expression.

\section{Local Minimax Risk and the Holevo Bound}
\label{app:holevo}

This appendix details the connection between the Local Asymptotic Normality (LAN) framework and the Holevo bound used in Theorem~\ref{thm:qcrb}. 

In the LAN limit \cite{vandervaart1998asymptotic, guta2009local, yamagata2013quantum}, the quantum statistical experiment of $n$ i.i.d.\ copies $\rho_\theta^{\otimes n}$ converges weakly to a quantum Gaussian shift model as $n \to \infty$. Let $\hat{\theta}_n$ be a regularized, locally unbiased estimator. Its rescaled covariance matrix $V = \lim_{n \to \infty} n \cdot \text{Cov}(\hat{\theta}_n)$ is fundamentally constrained by the Holevo bound \cite{holevo1982probabilistic}. 

For any positive definite weight matrix $W$, the weighted trace of the covariance matrix satisfies the matrix inequality:
\begin{equation}
\Tr(W V) \ge C_H[W],
\end{equation}
where $C_H[W]$ is the Holevo bound, defined as the minimum of $\Tr(W \Re Z) + \Tr(|W \Im Z|)$ over all valid estimator covariance matrices $Z$ arising from locally unbiased estimators $X = (X^i)$, where $Z_{ij} = \Tr(\rho X^i X^j)$. The condition $\Re Z \ge \Ical_{\mathrm{SLD}}^{-1}$ is a consequence of this definition, not the definition itself. 

In our setting, the local risk is asymptotically governed by the BKM metric, so we set the weight matrix to $W = g_{\mathrm{BKM}}(\theta_0)$. Combining the second-order expansion from Appendix~\ref{app:bkm-expansion} with the Holevo bound yields:
\begin{align}
\liminf_{n\to\infty} n \cdot R_n^*(\Vcal_Q, N_n^{(c)}) &= \liminf_{n\to\infty} \frac{n}{2} \E\left[ (\hat{\theta}_n - \theta)^T g_{\mathrm{BKM}}(\theta_0) (\hat{\theta}_n - \theta) \right] \nonumber \\
&= \frac{1}{2} \Tr\left( g_{\mathrm{BKM}}(\theta_0) V \right) \nonumber \\
&\ge \frac{1}{2} C_H\big[g_{\mathrm{BKM}}(\theta_0)\big].
\end{align}
Taking the limit $c \to \infty$ over the shrinking neighborhood $N_n^{(c)}$ establishes the Local Holevo-McCarthy Bound, proving that the asymptotic minimax risk is tightly bounded by the Holevo quantity weighted by the BKM metric.

\section{Phase Orbit Averaging and the Coherence Floor}
\label{app:coherence-proof}

We provide the step-by-step derivation for the proof of Proposition~\ref{prop:coherence-risk}. Let $\rho_\varphi = U_\varphi \rho U_\varphi^\dagger$, where $U_\varphi = \sum_k e^{i\varphi_k} \ket{k}\bra{k}$ is a diagonal unitary in the basis $B$. We integrate the relative entropy score over the normalized Haar measure $\dd\varphi$ of the phases $\varphi \in [0, 2\pi)^d$:
\begin{equation}
\int D(\rho_\varphi \| \hat{\sigma}) \dd\varphi = \int \Big( \Tr(\rho_\varphi \log \rho_\varphi) - \Tr(\rho_\varphi \log \hat{\sigma}) \Big) \dd\varphi.
\end{equation}
For the first term, von Neumann entropy is unitarily invariant, so $\Tr(\rho_\varphi \log \rho_\varphi) = \Tr(\rho \log \rho) = -S(\rho)$. Integrating this constant over the normalized measure yields simply $-S(\rho)$.

For the second term, we use the linearity of the trace and the integral:
\begin{equation}
-\int \Tr(\rho_\varphi \log \hat{\sigma}) \dd\varphi = -\Tr\left( \left( \int U_\varphi \rho U_\varphi^\dagger \dd\varphi \right) \log \hat{\sigma} \right).
\end{equation}
The group average of $\rho$ over the torus of diagonal unitaries is precisely the completely dephasing channel:
\begin{equation}
\int U_\varphi \rho U_\varphi^\dagger \dd\varphi = \sum_k \ket{k}\bra{k} \rho \ket{k}\bra{k} = \Deph_B(\rho).
\end{equation}
Thus, the second term evaluates to $-\Tr(\Deph_B(\rho) \log \hat{\sigma})$. 

We now add and subtract the von Neumann entropy of the dephased state, $S(\Deph_B(\rho)) = -\Tr(\Deph_B(\rho) \log \Deph_B(\rho))$:
\begin{align}
\int D(\rho_\varphi \| \hat{\sigma}) \dd\varphi &= -S(\rho) - \Tr(\Deph_B(\rho) \log \hat{\sigma}) \nonumber \\
&= \Big( S(\Deph_B(\rho)) - S(\rho) \Big) + \Big( -S(\Deph_B(\rho)) - \Tr(\Deph_B(\rho) \log \hat{\sigma}) \Big) \nonumber \\
&= C(\rho) + D(\Deph_B(\rho) \| \hat{\sigma}).
\end{align}
Since the quantum relative entropy is non-negative, $D(\Deph_B(\rho) \| \hat{\sigma}) \ge 0$, we obtain the pointwise lower bound:
\begin{equation}
\int D(\rho_\varphi \| \hat{\sigma}) \dd\varphi \ge C(\rho).
\end{equation}
Because the estimator $\hat{\sigma}$ is constructed exclusively from measurements in the basis $B$, its probability distribution is invariant under the phase rotations $U_\varphi$. Therefore, taking the expectation over the data $\E_{\mathrm{data}}$ preserves the inequality, completing the proof of the coherence floor.



\begin{thebibliography}{99}

\bibitem{mccarthy1956measures}
McCarthy, J. (1956). Measures of the Value of Information. \textit{Proceedings of the National Academy of Sciences}, 42(9), 654--655.

\bibitem{good1952rational}
Good, I. J. (1952). Rational Decisions. \textit{Journal of the Royal Statistical Society. Series B}, 14(1), 107--114.

\bibitem{savage1971elicitation}
Savage, L. J. (1971). Elicitation of Personal Probabilities and Expectations. \textit{Journal of the American Statistical Association}, 66(336), 783--801.

\bibitem{gneiting2007strictly}
Gneiting, T., \& Raftery, A. E. (2007). Strictly Proper Scoring Rules, Prediction, and Estimation. \textit{Journal of the American Statistical Association}, 102(477), 359--378.

\bibitem{nielsen2010quantum} 
Nielsen, M. A., \& Chuang, I. L. (2010). \textit{Quantum Computation and Quantum Information}. Cambridge University Press.

\bibitem{hayashi2006quantum} 
Hayashi, M. (2006). \textit{Quantum Information: An Introduction}. Springer.

\bibitem{wilde2017quantum}
Wilde, M. M. (2017). \textit{Quantum Information Theory} (2nd ed.). Cambridge University Press.

\bibitem{holevo1982probabilistic}
Holevo, A. S. (1982). \textit{Probabilistic and Statistical Aspects of Quantum Theory}. North-Holland.

\bibitem{guta2009local}
Kahn, J., \& Gu{\c{t}}{\u{a}}, M. (2009). Local asymptotic normality for finite dimensional quantum systems. \textit{Communications in Mathematical Physics}, 289(2), 597--652.

\bibitem{chitambar2019quantum}
Chitambar, E., \& Gour, G. (2019). Quantum resource theories. \textit{Reviews of Modern Physics}, 91, 025001.

\bibitem{frongillo2023quantum}
Frongillo, R. (2023). Quantum Information Elicitation. In \textit{Proceedings of the 24th ACM Conference on Economics and Computation (EC '23)} (p. 734). ACM.

\bibitem{quadeer2019minimax}
Quadeer, M., Tomamichel, M., \& Ferrie, C. (2019). Minimax quantum state estimation under Bregman divergence. \textit{Quantum}, 3, 126.

\bibitem{lewis1996convex}
Lewis, A. S. (1996). Convex analysis on the Hermitian matrices. \textit{SIAM Journal on Optimization}, 6(1), 164--177.

\bibitem{petz1986quasi} 
Petz, D. (1986). Quasi-entropies for finite quantum systems. \textit{Reports on Mathematical Physics}, 23(1), 57--65.

\bibitem{petz1996monotone}
Petz, D. (1996). Monotone metrics on matrix spaces. \textit{Linear Algebra and its Applications}, 244, 81--96.

\bibitem{hiai2011quantum}
Hiai, F., Mosonyi, M., Petz, D., \& Beny, C. (2011). Quantum $f$-divergences and error correction. \textit{Reviews in Mathematical Physics}, 23(7), 691--747.

\bibitem{wilde2018optimized}
Wilde, M. M. (2018). Optimized quantum $f$-divergences and data processing. \textit{Journal of Physics A}, 51(37), 374002.

\bibitem{degen2017quantum}
Degen, C. L., Reinhard, F., \& Cappellaro, P. (2017). Quantum sensing. \textit{Reviews of Modern Physics}, 89, 035002.

\bibitem{biamonte2017quantum}
Biamonte, J., Wittek, P., Pancotti, N., Rebentrost, P., Wiebe, N., \& Lloyd, S. (2017). Quantum machine learning. \textit{Nature}, 549(7671), 195--202.

\bibitem{daletskii1965integration} 
Daletskii, Y. L., \& Krein, S. G. (1965). Integration and differentiation of functions of Hermitian operators. \textit{AMS Translations (Series 2)}, 47, 1--30.

\bibitem{bhatia1997matrix}
Bhatia, R. (1997). \textit{Matrix Analysis}. Springer.

\bibitem{matsumoto2013new}
Matsumoto, K. (2013). A new quantum version of $f$-divergence. \textit{arXiv preprint arXiv:1311.4722}.

\bibitem{yamagata2013quantum}
Yamagata, K., Fujiwara, A., \& Gill, R. D. (2013). Quantum local asymptotic normality based on a new quantum likelihood ratio. \textit{The Annals of Statistics}, 41(4), 2197--2217.

\bibitem{vandervaart1998asymptotic}
van der Vaart, A. W. (1998). \textit{Asymptotic Statistics}. Cambridge University Press.

\bibitem{baumgratz2014quantifying} 
Baumgratz, T., Cramer, M., \& Plenio, M. B. (2014). Quantifying Coherence. \textit{Physical Review Letters}, 113(14), 140401.

\bibitem{winter2016operational}
Winter, A., \& Yang, D. (2016). Operational Resource Theory of Coherence. \textit{Physical Review Letters}, 116(12), 120404.

\bibitem{streltsov2017quantum}
Streltsov, A., Adesso, G., \& Plenio, M. B. (2017). Quantum Coherence as a Resource. \textit{Reviews of Modern Physics}, 89, 041003.

\bibitem{haah2017sample}
Haah, J., Harrow, A. W., Ji, Z., Wu, X., \& Yu, N. (2017). Sample-optimal tomography of quantum states. \textit{IEEE Transactions on Information Theory}, 63(9), 5628--5641.

\bibitem{chen2022adaptivity}
Chen, S., Huang, B., Li, J., Liu, A., \& Sellke, M. (2023). When Does Adaptivity Help for Quantum State Learning? In \textit{Proceedings of the 2023 IEEE 64th Annual Symposium on Foundations of Computer Science (FOCS)} (pp. 391--404). IEEE.

\bibitem{lowe2025lower}
Lowe, A., \& Nayak, A. (2025). Lower bounds for learning quantum states with single-copy measurements. \textit{ACM Transactions on Computation Theory}, 17(1), Article 7, 1--42.

\bibitem{odonnell2016efficient}
O'Donnell, R., \& Wright, J. (2016). Efficient quantum tomography. In \textit{Proceedings of the 48th Annual ACM SIGACT Symposium on Theory of Computing (STOC)} (pp. 899--912). ACM.

\end{thebibliography}
\end{document}